\documentclass[lettersize,journal]{IEEEtran}
\usepackage{color,array,amsthm}

\usepackage{amsmath, amsthm, amssymb}
\usepackage{csquotes}

\usepackage{hyperref}
\usepackage{float}
\usepackage{cite}
\usepackage{algorithm}
\usepackage{algorithmicx, algpseudocode}
\usepackage[table,xcdraw]{xcolor}
\usepackage{multirow}

\usepackage{mathtools}
\ifCLASSINFOpdf
\usepackage[pdftex]{graphicx}
\else
\fi
\usepackage{amsfonts}
\usepackage{amssymb}
\usepackage{color}
\usepackage{array}
\usepackage{fixltx2e}
\usepackage{stfloats}
\usepackage{url}
\newtheorem{theorem}{Theorem}
\newtheorem{lemma}{Lemma}
\newtheorem{corollary}{Corollary}

\usepackage{booktabs}
\usepackage{varwidth}
\usepackage[table,xcdraw]{xcolor}

\begin{document}

\title{Performance Analysis and Optimization of Multi-RIS-Aided UAV Networks}

\author{Khaled Alshehri$^1$, Anas M. Salhab$^2$, Ali Arshad Nasir$^3$

\thanks{$^1$Department of Mathematics and Statistics, King Fahd University of Petroleum and Minerals (KFUPM), Dhahran, Saudi Arabia (email: s201755390@kfupm.edu.sa);

$^2$Department of Electrical Engineering, King Fahd University of Petroleum and Minerals (KFUPM), Dhahran, Saudi Arabia (email: salhab@kfupm.edu.sa); 

$^3$Department of Electrical Engineering, King Fahd University of Petroleum and Minerals (KFUPM), Dhahran, Saudi Arabia (email: anasir@kfupm.edu.sa)}}

\maketitle
\begin{abstract}
In this paper, we study the performance of multiple reconfigurable intelligent surfaces (RISs)-aided unmanned aerial vehicle (UAV) communication networks over Nakagami-$m$ fading channels. For that purpose, we used accurate closed-form approximations for the channel distributions to derive closed-form approximations for the outage probability (OP), average symbol error probability (ASEP), and the average channel capacity assuming independent non-identically distributed (i.ni.d.) channels. Furthermore, we derive the asymptotic OP at the high signal-to-noise ratio (SNR) regime to get more insights into the system performance. We also study some practical scenarios related to RISs, UAV, and destination locations and illustrate their impact on the system performance through simulations. Finally, we provide an optimization problem on the transmit power of each channel.
\end{abstract}
\IEEEpeerreviewmaketitle
\section{INTRODUCTION}
Unmanned aerial vehicle (UAV) has received considerable attention as a solution for many wireless communication problems, $ e.g.$, networking problems in remote areas and locations of disasters. 
Given the need to improve the performance and coverage of UAV communication systems, many researchers have focused on reconfigurable intelligent surfaces (RISs) as a promising complementary solution. 
RIS is an artificial surface made of electromagnetic material capable of customizing the propagation of the radio waves impinging upon it \cite{Renzo1}- \cite{Alouini1}. It represents a new low-cost/less-complicated solution to realize wireless communication with high spectral and energy efficiencies.

Because of the promising potential of RIS technology, it has been recently used to enhance the performance of UAV networks in many studies and under different scenarios. In literature, there are various configurations for RIS-aided UAV systems; either a  RIS is installed on the UAV or left stationary. For example, one of the earliest papers on RIS-UAV systems is \cite{LiangdiRenzo}, where a system of a UAV that sends a signal to a receiver aided by stationary RISs were presented. The authors proposed an optimization algorithm to maximize the average achievable rate. In a more recent study, \cite{Yang4} considered a scenario where a source sends a signal to a UAV through a stationary RIS and then to a destination. The authors modeled the signal-to-noise ratio (SNR) statistics of the ground-to-air (G2A) link by considering the distribution of channel gains between the source and a RIS element to be Rayleigh, while the channel between the RIS element and UAV is Rician. For the air-to-ground (A2G) link, the authors modeled the SNR statistics with a Noncentral chi-squared distribution. Expressions of the outage probability and bit error rate were derived based on the approximations of the system SNR statistics. The average channel capacity was also derived using a tight upper bound.

Another example is the work done by the authors in \cite{Agrawal}, where they considered a UAV communicating with a vehicle through a RIS with multiple interfering vehicles operating in the same frequency spectrum. The authors analyzed the system performance for finite blocklength (FBL) and infinite BL-based transmissions. They assumed the RIS channel gains follow a Nakagami-$m$ distribution. \cite{Zhangs} is one of the earliest papers that treated the case where RIS is installed on a UAV. The authors proposed a Q-learning algorithm to optimize the downlink communication capacity.
 In a recent paper \cite{YangLi}, the authors developed two different distributions to model the system's SNR statistical distribution.
 
Although previous studies did provide important insights into UAV and RIS-aided UAV systems, the multi-RIS-aided UAV networks were not properly studied, especially under the multi-RIS model in \cite{Salhab} where the authors analyzed the statistics of an opportunistic RIS selection model using the Laguerre series approach. Using multi-RIS instead of a single RIS increases the system diversity order when the hop where RISs are used dominates the system performance.

Motivated by the above observations, we summarize the contribution of this paper as follows: 
\begin{itemize}
\item We provide a performance analysis of multi-RIS-aided UAV networks over Nakagami-$m$ channels for the source-to-RISs links and Rician fading channel for the UAV-to-destination fading link. To cover more scenarios and to have flexibility in RIS locations between the source and UAV, an independent and non-identically distributed (i.ni.d.) case is considered for RIS channels.

\item We derive accurate approximations for the system outage probability (OP), average symbol error probability (ASEP), and average channel capacity. The derived expressions are valid for an arbitrary number of reflecting elements and non-integer values of the Nakagami fading parameter. 

\item We further study the effect of various system parameters such as the number of reflecting elements, number of RISs, distances between the source to RIS and RIS to UAV, location of RISs, UAV's height and horizontal distance, etc., on the system performance. They are investigated under realistic conditions, as shown in the Numerical Results Section.

\item We also formulate and solve an optimization problem on the transmit power of each channel. Given a total transmit power $E_{T}$, what is the optimal transmit power for the first and second channels that minimize the total outage probability?
\end{itemize}

\subsection{Paper Organization, Notations, and Symbols}
The rest of this paper is organized as follows. Section \ref{SCMs} presents the system and channel models. The performance analysis is evaluated in Section \ref{EPA}. The optimization problem and its solution are presented in Section \ref{optimization}. Simulations and numerical results are discussed in Section \ref{Numerical Results}. Finally, the paper is concluded in Section \ref{C}.

\textit{Notations and symbols:}
The functions and operators used throughout this paper are in Table.\ref{TEST1}.

	\begin{table}[]
		\centering
		\caption{Notations used in the paper.}
		\label{TEST1}
  \scalebox{0.95}{
		\begin{tabular}{|l|l|}	
			\hline
			\textbf{Notation}          & \textbf{Definition}                                        \\ \hline
			$ P_{r}\left[\cdot\right] $           & Probability operator                            \\ \hline
			$F_{X}\left(x\right)$      & Cumulative distribution function of a random variable $ X $ \\ \hline
			$f_{X}\left(x\right)$      & Probability density function of a random variable $ X $    \\ \hline
			$|\cdot|$                  & Absolute value                                              \\ \hline
			$\mathbb{C}^{m\times n}$  & Set of matrices with dimension $m\times n$                    \\ \hline
			
			$(\cdot)^{T}$              & Transpose operator                                            \\ \hline
			$\mathbb{E}\left[\cdot\right]$               & Expectation operator                                              \\ \hline
			${Var}\left(\cdot\right)$               & Variance operator                                              \\ \hline
			$\gamma(\cdot,\cdot)$        & Lower incomplete Gamma function                                   \\ \hline
			$\Gamma(\cdot)$            & Gamma function                                                    \\ \hline
			$\exp(\cdot)$             & Exponential function    \\ \hline
			$ Q(\cdot) $             &   Q-function    \\ \hline
		    $L_n^{(\alpha)}(\cdot)$           & Generalized Laguerre polynomials                             \\ \hline
		    $Q_v\left(\cdot,\cdot\right)$           & Generalized Marcum Q-function                                                 \\ \hline
		    ${ }_{p} F_{q}\left(\cdot ; \cdot ; \cdot\right)$           & Generalized hypergeometric function \\ \hline
		    $I_{v}\left(\cdot\right)$           & Modified Bessel function of the first kind
      \\ \hline 
		\end{tabular} }
	\end{table} 

\section{SYSTEM AND CHANNEL MODELS}\label{SCMs}
Consider a multi-RIS network that assists the communication between a source and the UAV in the first hop. Only one among the multiple RISs is selected for communication. During the second hop, RIS forwards the decoded version of the signal to the destination.

This section illustrates the proposed system and channel models of multi-RIS-assisted UAV systems with the RIS selection strategy.

\subsection{Multiple RISs-Assisted G2A Channel}
We consider a multi-RIS assisted network, in which a single-antenna source (S) communicates with a UAV assisted by $K$ multiple RISs on the ground. Each RIS, $\left\{\text{RIS}_{k}\right\}_{k=1}^{K}$, is equipped with $N_{k}$ passive elements. We also assume that the communication between the S and UAV is only provided via RISs and the direct communication links between them are unavailable due to natural or man-made obstacles. The S-RIS$_{k}$ links in the first hop are assumed to undergo Nakagami-\emph{m} fading with shape parameter $m_{k,1}$ and scale parameter $\Omega_{k,1}$, where the channel vector between the S and RIS$_{k}$ is denoted by $\mathbf{h}_{k}\in \mathbb{C}^{N_{k}\times1}$, $\mathbf{h}_{k} = \left[h^{(1)}_{k}, ..., h^{(i)}_{k}, ..., h^{(N_{k})}_{k}\right]^{T}$, $h^{(i)}_{k}=\frac{1}{\sqrt{P_{L,k,1}}}\alpha^{(i)}_{k}e^{-j\phi^{(i)}_{k}}$ denotes the channel coefficient between the S and RIS$_{k}$ $i$th element, where $P_{L,k,1}$, $\alpha^{(i)}_{k}$, and $\phi^{(i)}_{k}$,respectively refer to the path-loss, channel amplitude, and channel phase of for the first hop. Likewise, the RIS$_{k}$-UAV links are also assumed to have Nakagami-\emph{m} fading channel model with shape parameter $m_{k,2}$ and scale parameter $\Omega_{k,2}$, where the channel vector between the RIS$_{k}$ and UAV is denoted by $\mathbf{g}_{k}\in \mathbb{C}^{N_{k}\times1}$, $\mathbf{g}_{k} = \left[g^{(1)}_{k}, ..., g^{(i)}_{k}, ..., g^{(N_{k})}_{k}\right]^{T}$, $g^{(i)}_{k} = \frac{1}{\sqrt{P_{L,k,2}}}\beta^{(i)}_{k}e^{-j\Phi^{(i)}_{k}}$ denotes the channel coefficient between $i$th RIS element and UAV, where $P_{L,k,2}$, $\beta^{(i)}_{k}$, and $\Phi^{(i)}_{k}$, respectively refer to the path-loss, channel amplitude, and channel phase of the second hop. The reflection coefficients of the RIS$_{k}$ are denoted by the entries of the diagonal matrix $\mathbf{\Theta}_{k}\in \mathbb{C}^{N_{k}\times N_{k}}$, for the $i$th element. Under the full reflection assumption, we have $\Theta^{(i,i)}_{k}=e^{j\theta^{(i)}_{k}}$, where $\theta^{(i)}_{k}\in[0,2\pi)$. 
The signal received at UAV from the reflected signals of RIS$_{k}$ can be expressed as
\begin{align} 
\label{receivedsignal}
y_{k} &= \sqrt{\frac{E_{s}}{P_{L,k}}}\sum_{i = 1}^{N_{k}}\alpha^{(i)}_{k}\beta^{(i)}_{k}s + n_{k},
\end{align}
where $E_{s}$ is the average power of the transmitted signal, $s$ is the transmitted signal, ${n}_{k}$ denotes the additive white Gaussian noise (AWGN) sample with zero mean and variance $N_{0}$, and $P_{L,k}$ denotes the overall path-loss of the RIS$_{k}$-assisted path. Here, $ P_{L,k} = P_{L,k,1}P_{L,k,2} $ and given by \cite{Basar}
\begin{align} 
\label{path-loss}
P_{L,k} &= \left(\left(\frac{\lambda}{4\pi}\right)^{4}\frac{G_{k,1}G_{k,2}}{d^{2}_{k,1}d^{2}_{k,2}}\epsilon_{k}\right)^{-1},
\end{align}
where $\lambda$ is the wavelength, $ G_{k,1}$ and $G_{k,2}$ are the gains of the RIS$_{k}$ in the first and second hops, respectively, and $\epsilon_{k}$ is the efficiency of RIS$_{k}$, which is described as ratio of transmitted signal power by RIS to received signal power by RIS. In this paper, it is assumed that $\epsilon_{k} = 1$. Furthrmore, $d_{k,1}$ and $d_{k,2}$ are the distances from the source-to-RIS$_{k}$ and RIS$_{k}$-to-UAV, respectively.

The RIS$_{k}$ optimizes the phase reflection coefficients to maximize the received SNR at UAV, by aligning the phases of the reflected signals to the sum of the phases of its incoming and outgoing fading channels. Thus, the maximized e2e SNR for RIS$_{k}$ can be expressed as \cite{Alouini1}
\begin{equation} \label{eq:gamma}
\gamma_{k} =\frac{E_{s}}{N_{0}P_{L,k}}\left ({\sum _{i=1}^{N_{k}} \alpha^{(i)}_{k}\beta^{(i)}_{k} }\right)^{2}= \frac{\overline{\gamma}}{P_{L,k}}Z^{2}_{k},
\end{equation}
where $\overline{\gamma}=\frac{E_{s}}{N_{0}}$ is the average SNR. 

\subsection{RIS Selection Strategy}\label{sec:RIS_sel_st}
In this system, we consider that one out of $K$ RISs is selected to aid the communications. Specifically, the choice of the suitable RIS is performed to maximize the received signal at the UAV. Therefore, the maximized end-to-end(e2e) SNR of the selected RIS can be expressed as
\begin{equation}\label{eq:RISselection}
	\gamma_a=\underset{k= 1,...,K}{{\max}}\{\gamma_{k}\}. 
\end{equation}
Let $a_{k}=\frac{m_{k,1} m_{k,2} N_{k} \Gamma (m_{k,1})^2 \Gamma (m_{k,2})^2}{m_{k,1} m_{k,2} \Gamma (m_{k,1})^2 \Gamma (m_{k,2})^2-\Gamma \left(m_{k,1}+\frac{1}{2}\right)^2 \Gamma \left(m_{k,2}+\frac{1}{2}\right)^2}-N_{k},$ and \begin{small} $b_{k}=\frac{m_{k,1} m_{k,2} \Gamma (m_{k,1})^2 \Gamma (m_{k,2})^2-\Gamma \left(m_{k,1}+\frac{1}{2}\right)^2 \Gamma \left(m_{k,2}+\frac{1}{2}\right)^2}{\sqrt{\frac{m_{k,1}}{\Omega_{k,1}}} \Gamma (m_{k,1}) \Gamma \left(m_{k,1}+\frac{1}{2}\right) \sqrt{\frac{m_{k,2}}{\Omega_{k,2}}} \Gamma (m_{k,2}) \Gamma \left(m_{k,2}+\frac{1}{2}\right)}$.\end{small}
Then The CDF and PDF of $\gamma_{a}$ can be, respectively given by \cite{Salhab}  \begin{small}
\begin{equation}\label{eq:CDFRISselectionsimple}
 F_{\gamma_{a}}(\gamma) = \sum_{n_{1} = 0}^{\infty}...\sum_{n_{K} = 0}^{\infty} \prod_{k=1}^{K}\frac{\left(-1\right)^{n_{k}}\left(\sqrt{\frac{P_{L,k} }{\overline{\gamma_{a}}b^{2}_{k}}}\right)^{a_{k} + n_{k}}}{n_{k}!\left(a_{k} + n_{k}\right)\Gamma (a_{k})} \gamma^{\frac{\sum_{k=1}^{K}\left(a_{k} + n_{k}\right)}{2}},
\end{equation}
and 
\begin{align}\label{eq:PDFRISselection}
\nonumber f_{\gamma_{a}}(\gamma) &= \sum_{j = 1}^{K}\frac{\left(\frac{ P_{L,j} }{\overline{\gamma}}\gamma\right)^{\frac{a_{j}-1}{2}} \exp \left(-\sqrt{\frac{ P_{L,j}}{\overline{\gamma_{a}}b^{2}_{j}}\gamma}\right)}{2 b^{a_{j}}_{j} \Gamma (a_{j}) \sqrt{\frac{\overline{\gamma_{a}} }{P_{L,j}}\gamma }}\sum_{n_{k}\ne j, n_{k}=1}^{\infty}...\sum_{n_{k}\ne j, n_{k}=K}^{\infty} \\ &\prod_{n_{k}\ne j, n_{k}=1}^{K}\frac{\left(-1\right)^{n_{k}}\left(\sqrt{\frac{P_{L,k} }{\overline{\gamma_{a}}b^{2}_{k}}}\right)^{a_{k} + n_{k}}}{n_{k}!\left(a_{k} + n_{k}\right)\Gamma (a_{k})} \gamma^{\frac{\sum_{k\ne j, k=1}^{K}\left(a_{k} + n_{k}\right)}{2}}.
\end{align} \end{small}

\subsection{A2G Channel}\label{sec:RIS_sel_st}
In the second channel, the UAV sends the signal to a destination on ground. The signal received at the destination (D) is given by \begin{equation}
y_{b}=h_0x+n_u,
\end{equation}
where $h_0=\frac{1}{\sqrt{L}} \chi e^{-j \theta}$ with $\theta$ being the phase of the second channel and $n_{u} {\sim}  \mathcal{C N}(0,N_{u})$ denotes the additive white Gaussian noise. Then, the instantaneous SNR at D is \begin{equation}
\gamma_{b}=\frac{|\chi|^{2} E_{u}}{N_{u} L},
\end{equation}
where $L=10 \cdot \log _{10}\left(L_{0}^{\alpha}\right)+A$ is the path-loss, $E_{u}$ is the transmit power of the UAV, $L_{0}=\sqrt{h^{2}+r_{0}^{2}}$ is the distance from UAV to D, where $r_{0}$ and $h$ are the UAV's horizontal distance from D and height, respectively. $\alpha$ is the path-loss exponent and it is given by $\alpha=a_{1} P_{L o S}+b_{1}$, where $a_{1} , b_{1}$ are constants determined by the environment and the transmit frequency and $P_{LoS}$ is the line of sight(LoS) probability. $\chi$ is the channel amplitude, which follows a Rician distribution. Then, the PDF of $\gamma_{b}$ can be written as 
\begin{align}
f_{\gamma_{b}}(\gamma)&=\frac{(1+K_{0}) e^{-K_{0}} L}{\overline{\gamma_{b}}} \exp \left(-\frac{1+K_{0}}{\overline{\gamma_{b}}} L \gamma\right) \\ \notag &\times I_{0}\left(2 \sqrt{\frac{K_{0}(1+K_{0})}{\overline{\gamma_{b}}} L \gamma}\right),
\end{align} 
where $\overline{\gamma_{b}}=\frac{E_{u}}{N_{u}}$ , $K_{0}$ is the Rician factor. The CDF of $\gamma_{b}$ is given by \cite{Azari}
\begin{equation}\label{eq:18}
F_{\gamma_{b}}(\gamma)=1-Q_{1}\left(\sqrt{2 K_{0}}, \sqrt{\frac{2 \gamma L(1+K_{0})}{\overline{\gamma_{b}}}}\right).
\end{equation}

The LoS probability $P_{LoS}$ is modeled as follows \cite{Hourani}:
\begin{equation}
P_{L o S}=\prod_{n=0}^{m= \lfloor  \frac{r_{0}}{200} \sqrt{6}-1\rfloor}\left[1-\exp \left(-h^{2} \cdot \frac{\left(1-\frac{n+\frac{1}{2}}{m+1}\right)^{2}}{450}\right)\right]
.\end{equation} Also, from \cite{Azari}, we model the Rician factor as $K_{0}=a_{2} \cdot e^{b_{2} \tan ^{-1}\left(\frac{h}{r_{0}}\right)},$ where $a_{2}$ and $b_{2}$ are constants to be determined later.
\section{PERFORMANCE ANALYSIS}\label{EPA}
In this section, we derive theoretical expressions for the system outage probability and average symbol error probability. 
\subsection{Outage Probability}
The OP of each link is given by
\begin{equation}\label{eq:Pout}
P_{i,\text{out}}=F_{\gamma_{i}}(\gamma_\text{out}).
\end{equation}
Where $i \in \{a,b\}.$

We measure the total system outage performance by considering the probability that $\min \left\{\gamma_{a}, \gamma_{b}\right\}$ falls below a predetermined threshold value. Therefore, The total OP is given by \begin{equation}\label{eq:21}
P_\text{out}=P_{a,\text{out}}+P_{b,\text{out}}-P_{a,\text{out}} P_{b,\text{out}}
\end{equation}

\begin{corollary} The asymptotic expression of the OP, when $\min \left\{\overline{\gamma_{a}}, \overline{\gamma_{b}}\right\} \rightarrow \infty$, is given by
\begin{align}\label{outasym}
   \notag P_\text{out}^{\infty}&= \prod_{k=1}^{K}\left(\frac{b_{k}^{2}}{\gamma_\text{out} P_{L, k} \Gamma\left(a_{k}\right)^{-\frac{2}{a_{k}}}} \overline{\gamma_{a}}\right)^{-\frac{a_{k}}{2}} \\&+\frac{e^{-K_{0}}\left(1+K_{0}\right) \gamma_\text{out}L}{\overline{\gamma_{b}}}
\end{align}
\end{corollary}
\begin{proof}
by \eqref{eq:21} we have : 
\begin{equation}
P_\text{out}^{\infty}=P_{a,\text{out}}^{\infty}+P_{b,\text{out}}^{\infty}-P_{a,\text{out}}^{\infty} P_{b,\text{out}}^{\infty} \rightarrow P_{a,\text{out}}^{\infty}+P_{b,\text{out}}^{\infty}\end{equation}
The asymptotic OP for the first channel can be calculated by $P^{\infty}_{a,\text{out}}=F^{\infty}_{\gamma_{a}}(\gamma_\text{out})=\prod_{k = 1}^{K}F^{\infty}_{\gamma_{k}}(\gamma_\text{out})$. Then, using \cite[Eq. (8.354.1)]{Grad.}, we get
$F^{\infty}_{\gamma_{k}}(\gamma_\text{out})= \sum_{n_{k}=0}^{\infty}\frac{(-1)^{n_{k}}\left(\sqrt{\frac{\gamma_\text{out}}{\frac{\overline{\gamma_{a}}}{P_{L,k}}}}\right)^{a_{k}+n_{k}}}{(a_{k}+n_{k})b^{a_{k}+n_{k}}_{k}\Gamma(a_{k})} $. As $\overline{\gamma_{a}}\rightarrow\infty$, this expression is only dominated by the first term in summation. Upon considering that, we get $P^{\infty}_{a,\text{out}}$.
For the second channel, the Laguerre series of the OP \cite[Eq. (8)]{Andras}: $P_{b, \text{out}}=e^{-K_{0}} \sum_{i=0}^{\infty}(-1)^{K_{0}} \frac{L_{i}^{0}\left(K_{0}\right)}{\Gamma(i+2)}\left(\frac{\gamma_\text{out}\left(1+K_{0}\right)L}{\overline{\gamma_{b}}}\right)^{i+1}$, by considering the first term the proof is completed.
\end{proof}

\subsection{Average Symbol Error Probability}
\begin{theorem}
The ASEP for each channel can be given by
\begin{align}\label{eq:ASEPclosedform}
\notag P_{a}&=\frac{p\sqrt{q}}{2\sqrt{\pi}}\sum_{n_{1} = 0}^{\infty}...\sum_{n_{K} = 0}^{\infty} \prod_{k=1}^{K}\frac{\left(-1\right)^{n_{k}}\left(\sqrt{\frac{P_{L,k} }{\overline{\gamma_{a}}b^{2}_{k}}}\right)^{a_{k} + n_{k}}}{n_{k}!\left(a_{k} + n_{k}\right)\Gamma (a_{k})} \\ &\times\frac{1}{q^{\frac{\sum_{k=1}^{K}\left(a_{k} + n_{k}\right)}{2} + \frac{1}{2}}}\Gamma\left(\frac{\sum_{k=1}^{K}\left(a_{k} + n_{k}\right)}{2}+\frac{1}{2}\right),
\end{align}
\begin{align}\label{eq:pb}
\notag P_{b}&=\frac{p \sqrt{q}}{2} \sqrt{\frac{\frac{\overline{\gamma}_{b}}{L}}{q \frac{\overline{\gamma}_{b}}{L}+K_{0}+1}} e^{-K_{0}} \\ \notag&\times\Biggl[\Phi_{1}\left(\frac{1}{2}, 1,1 ; \frac{K_{0}+1}{q \frac{\overline{\gamma}_{b}}{L}+K_{0}+1}, \frac{K_{0}\left(K_{0}+1\right)}{q \frac{\overline{\gamma}_{b}}{L}+K_{0}+1}\right)\\  &-{ }_{1} F_{1}\left(\frac{1}{2} ; 1 ; \frac{K_{0}\left(K_{0}+1\right)}{q \frac{\overline{\gamma}_{b}}{L}+K_{0}+1}\right) \Biggr],
\end{align}
where $p$ and $q$ are constants representing the type of modulation. 
\end{theorem}
\begin{proof} 
Our results apply for all general modulation formats that have an ASEP expression of the form	
\begin{align}\label{eq:ASEPgeneral}
P_{i}=E_{\gamma_{i}}\left[pQ(\sqrt{2q\gamma})\right],
\end{align} 
where $ Q(\cdot) $ is the Gaussian Q-function. Such modulation formats include binary phase-shift keying (BPSK) ($ p = 1, q = 1 $) and M-ary PSK ($ p = 2, q = \sin\left(\frac{2\pi}{M}\right) $).		
The ASEP can be obtained as follows \cite{McKay}
\begin{align}\label{eq:ASEP}
P_{i}=\frac{p\sqrt{q}}{2\sqrt{\pi}}\int_{0}^{\infty}\frac{\exp\left(-q\gamma\right)}{\sqrt{\gamma}}F_{\gamma_{i}}(\gamma) d\gamma.
\end{align}
For $P_{a}$, inserting \eqref{eq:CDFRISselectionsimple} in \eqref{eq:ASEP} gives
\begin{align}\label{eq:ASEP1}
\notag P_{a}&=\frac{p\sqrt{q}}{2\sqrt{\pi}}\sum_{n_{1} = 0}^{\infty}...\sum_{n_{K} = 0}^{\infty} \prod_{k=1}^{K}\frac{\left(-1\right)^{n_{k}}\left(\sqrt{\frac{P_{L,k} }{\overline{\gamma}b^{2}_{k}}}\right)^{a_{k} + n_{k}}}{n_{k}!\left(a_{k} + n_{k}\right)\Gamma (a_{k})} \\  &\times\int_{0}^{\infty} \gamma^{\frac{\sum_{k=1}^{K}\left(a_{k} + n_{k}\right)}{2} -\frac{1}{2}} \exp\left(-q\gamma\right) d\gamma.
\end{align}
With the help of $\int_{0}^{\infty}x^{v-1}\exp\left(-\mu x\right) dx = \frac{1}{\mu^{v}}\Gamma\left(v\right)$ \cite[Eq. (3.381.4)]{Grad.}, we obtain \eqref{eq:ASEPclosedform}. 

For $P_{b}$ , inserting \eqref{eq:18} in \eqref{eq:ASEP} and using \cite[Eq. (3)]{sharif} we obtain \eqref{eq:pb} and the proof is completed. 

\end{proof}
 For a DF relaying, dual-hop communication system, the ASEP is given by \cite{salhab2} \begin{equation}\label{eq:dh} P_{e}=P_{a}+P_{b}-2 P_{a} P_{b}
\end{equation}

Upon substituting \eqref{eq:ASEPclosedform} \& \eqref{eq:pb} into \eqref{eq:dh}, we get the overall system ASEP.
\subsection{Average Channel Capacity}
The average channel capacity can be expressed as $$
C=\frac{1}{\ln (2)} \int_0^{\infty} \ln (1+\gamma) f_\gamma(\gamma) d \gamma
.$$
With some manipulations the average channel capacity can be expressed in terms of the CDF as \begin{equation}\label{I1}
C=\frac{1}{\ln (2)} \int_0^{\infty} \frac{1-F_\gamma(\gamma)}{1+\gamma} d \gamma
.\end{equation}
For two continuous and unbounded functions $f,g$, we mean by $f \approx g$ is that $\lim _{x \rightarrow \infty}\left|f(x)-g(x)\right|=0 .$ Then we have the following lemma.
\begin{lemma}
Let $F_\gamma(\gamma, \overline{\gamma})$ be the CDF of $\gamma$ with $
\mathbb{E}\left[\gamma^2\right]=\overline{\gamma}^2 E
$ for some constant $E$ and dependent on a parameter $\overline{\gamma}$.
Then, \begin{equation}\label{I2}
C(\overline{\gamma}) \approx \frac{1}{\ln (2)}\int_0^{\overline{\gamma}^2} \frac{1-F_\gamma(\gamma, \overline{\gamma})}{1+\gamma} d \gamma
.\end{equation}
\end{lemma}
\begin{proof}
We will show that $$
\lim _{\overline\gamma \rightarrow \infty}\left|\int_0^{\overline{\gamma}^2} \frac{1-F_\gamma(t, \overline{\gamma})}{1+t} d t-\int_0^{\infty} \frac{1-F_\gamma(t, \overline{\gamma})}{1+t} d t\right|=0
.$$ That is $$
\lim _{\overline\gamma \rightarrow \infty}\left|\int_{\overline{\gamma}^2}^{\infty} \frac{1-F_\gamma(t, \overline{\gamma})}{1+t} d t\right|=0
.$$
By Markov's inequality, $\mathrm{P}(|\gamma| \geq t) \leq \frac{\mathbb{E}\left(|\gamma|^2\right)}{t^2}$ so, \\$1-F_\gamma(t, \overline{\gamma}) \leq\frac{\mathbb{E}\left(|\gamma|^2\right)}{t^2}$ that is $\frac{1-F_\gamma(t, \overline{\gamma})}{1+t} \leq \frac{\mathbb{E}\left[\gamma^2\right]}{(1+t) t^2}$.
But, $$
\mathbb{E}\left[\gamma^2\right]=\overline{\gamma}^2 E
.$$ So, $\frac{1-F_\gamma(t, \overline{\gamma})}{1+t} \leq \frac{\overline{\gamma}^2 E}{(1+t) t^2}$. By taking the integral on both sides we have, \begin{align}
\notag \int_{\overline{\gamma}^2}^{\infty} \frac{1-F_\gamma(t, \overline{\gamma})}{1+t} d t &\leq \int_{\overline{\gamma}^2}^{\infty} \frac{\overline{\gamma}^2 E}{(1+t) t^2} d t \\ \notag &= \left(1-\overline{\gamma}^2 \ln \left(1+\frac{1}{\overline{\gamma}^2}\right)\right)E
.\end{align} But, $\lim _{\overline\gamma \rightarrow \infty} \left(1-\overline{\gamma}^2 \ln \left(1+\frac{1}{\overline{\gamma}^2}\right)\right)E=0$, which completes the proof.
\end{proof}

It is straightforward to show that if $\phi(\overline{\gamma}) \geq \overline{\gamma}^{2}$, then $C(\overline{\gamma}) \approx \frac{1}{\ln (2)}\int_0^{\phi(\overline{\gamma})} \frac{1-F_\gamma(\gamma, \overline{\gamma})}{1+\gamma} d \gamma
.$ This can be useful for a better approximation and faster asymptotic convergence.
\begin{theorem}
The average channel capacity of the G2A channel is given by \begin{align}
\notag C_{1}(\overline{\gamma_a})&=\frac{1}{\ln (2)} \lim _{\gamma \rightarrow \infty}[\ln (1+\gamma)-\sum_{n_1=0}^{\infty} \ldots \sum_{n_k=0}^{\infty}\\ \notag  &\prod_{k=1}^K \frac{(-1)^{n_k}\left(\sqrt{\frac{P_{L, k}}{\overline{\gamma_a} b_k^2}}\right)^{n_k+a_k}}{n_{k} !\left(a_k+n_k\right) \Gamma\left(a_k\right)}\left(\frac{\gamma^{r+1}}{r+1}\right) \\ &\times{}_{2}F_{1}\left(1, r+1 ; r+2 ;-\gamma\right)],\end{align}
\begin{align}
   \notag C_{1}(\overline{\gamma_a}) &\approx \frac{1}{\ln (2)}[\ln (1+\overline{\gamma_a}^{2})-\sum_{n_1=0}^{\infty} \ldots \sum_{n_k=0}^{\infty}\\ \notag &\prod_{k=1}^K \frac{(-1)^{n_k}\left(\sqrt{\frac{P_{L, k}}{\overline{\gamma_a} b_k^2}}\right)^{n_k+a_k}}{n_{k} !\left(a_k+n_k\right) \Gamma\left(a_k\right)}\left(\frac{(\overline{\gamma_a}^{2})^{r+1}}{r+1}\right) \\ &\times{}_{2}F_{1}\left(1, r+1 ; r+2 ;-\overline{\gamma_a}^{2}\right)]
,\end{align}
where $r=\frac{\sum_{k=1}^K\left(n_k+a_k\right)}{2}.$
\end{theorem}
\begin{proof}
We have $
\gamma_{a}=\max \left(\gamma_k\right)=\max \left(\overline{\gamma_{a}} Z_k\right)=\overline{\gamma_{a}} \max \left(Z_k\right), \mathbb{E}\left[\gamma_{a}^2\right]=\overline{\gamma_{a}}^2 E
.$Then, by substituting the CDF of the first hop directly in the integrals \eqref{I1}\&\eqref{I2} and using $\int_0^u \frac{x^{\mu-1} d x}{1+\beta x}=\frac{u^\mu}{\mu}{ }_2 F_1(1, \mu ; 1+\mu ;-\beta u)$\cite[Eq. (3.195.5)]{Grad.}, the proof is complete.
\end{proof}
We can simplify the previous expression further.
\begin{corollary}
If $\sum_{k=1}^K a_k$ is not an integer, then the average channel capacity of the G2A channel is given by: \begin{align}
\notag C_{1}(\overline{\gamma_a})&=\frac{1}{\ln (2)} \lim _{\gamma \rightarrow \infty}[\ln (1+\gamma)-\sum_{n_1=0}^{\infty} \ldots \sum_{n_k=0}^{\infty}\\ \notag &\prod_{k=1}^K \frac{(-1)^{n_k}\left(\sqrt{\frac{P_{L, k}}{\overline{\gamma_a} b_k^2}}\right)^{n_k+a_k}}{n_{k} !\left(a_k+n_k\right) \Gamma\left(a_k\right)}(\sum_{n_{k+1}=0}^{\infty} [\frac{(-1)^{n_{k+1}}}{r-n_{k+1}} \gamma^{r-n_{k+1}}]\\ &-\pi \csc (\pi r))].\end{align} 
\begin{align}
   \notag C_{1}(\overline{\gamma_a}) &\approx \frac{1}{\ln (2)}[\ln (1+\overline{\gamma_a}^{2})-\sum_{n_1=0}^{\infty} \ldots \sum_{n_k=0}^{\infty}\\ \notag &\prod_{k=1}^K \frac{(-1)^{n_k}\left(\sqrt{\frac{P_{L, k}}{\overline{\gamma_a} b_k^2}}\right)^{n_k+a_k}}{n_{k} !\left(a_k+n_k\right) \Gamma\left(a_k\right)}(\sum_{n_{k+1}=0}^{\infty} [\frac{(-1)^{n_{k+1}}}{r-n_{k+1}} \overline{\gamma_a}^{r-n_{k+1}}]\\ &-\pi \csc (\pi r))]
.\end{align} 
where $r=\frac{\sum_{k=1}^K\left(n_k+a_k\right)}{2}.$
\end{corollary}
\begin{proof}
See Appendix A.
\end{proof}
The assumption in this corollary is reasonable since $\sum_{k=1}^K a_k$ is almost always not an integer.
\begin{theorem}
The average channel capacity of the A2G channel is given by:\begin{align}
\notag C_{2}\left(\overline{\gamma_b}\right)&=\frac{1}{\ln (2)} \lim _{\gamma \rightarrow \infty}\small[\ln (1+\gamma) R_1\\ &+\sum_{k=0}^{\infty} \sum_{i=1}^{k+1} \frac{(-1)^i e^{-K_0} L_k\left(K_0\right)\left(\frac{1+K_0}{\overline{\gamma_b}}\right)^{k+1}}{\Gamma(k+2) i} \gamma^i\small]
,\end{align}
\begin{align}
\notag C_{2}\left(\overline{\gamma_b}\right) &\approx \frac{1}{\ln (2)}\small[\ln \left(1+{\overline{\gamma_b}}^2\right) R_2\\&+\sum_{k=0}^{\infty} \sum_{i=1}^{k+1} \frac{(-1)^i e^{-K_0} L_k\left(K_0\right)\left(1+K_0\right)^{k+1}}{\Gamma(k+2) i} \overline{\gamma}_b^{2 i-k-1}\small]
.\end{align}
where $R_1=1+\sum_{k=0}^{\infty} \frac{e^{-K_0} L_k\left(K_0\right)\left(\frac{1+K_0}{\overline{\gamma_b}}\right)^{k+1}}{\Gamma(k+2)}$ and $R_2=1+\frac{e^{-K_0}\left(1+K_0\right)}{2 \overline{\gamma_b}}.$
\end{theorem}
\begin{proof}
We have $
 \mathbb{E}\left[\gamma_{b}^2\right]=\overline{\gamma_{b}}^2 E
.$ Then, by substituting the CDF of the second hop directly in the integrals \eqref{I1}\&\eqref{I2} and using $Q_v(a, b)=1-\sum_{n \geq 0}(-1)^n e^{-\frac{a^2}{2}} \frac{L_n^{(v-1)}\left(\frac{a^2}{2}\right)}{\Gamma(v+n+1)}\left(\frac{b^2}{2}\right)^{n+v}$\cite[Eq. (8)]{Andras}, with\\$\int_0^t \frac{x^{k+1}}{1+x} d \gamma=(-1)^{k-1}\left[\sum_{i=1}^{k+1} \frac{(-1)^i t^i}{i}\right]+(-1)^{k+1} \ln (t+1)$ the proof is complete.
\end{proof}

The average channel capacity of the system is given by \cite{salhab2}
\begin{equation}\label{eq:totalc}
C=\frac{1}{2} \min [ C_{1}(\overline{\gamma_a}) , C_{2}(\overline{\gamma_b})].
\end{equation}
Upon substituting the expressions in the previous results into \eqref{eq:totalc}, we get the overall system average channel capacity.
\section{OPTIMIZATION}\label{optimization}
In this section, we solve the following optimization problem: 

Given a total transmit power budget $E_{T}$, what is the optimal transmit power for the first and second channels that will minimize the total outage probability?
 More precisely our problem can be formulated as follows:
 $$
\begin{array}{ll}
{\operatorname{minimize}} & P_\text{out}^{\infty}(E_{s},E_{u}), \\
\text { subject to } & E_{s}+E_{u} \leq E_{T},
\end{array}$$
 where $P_\text{out}^{\infty}$ is the asymptotic outage probability in \eqref{outasym}.
Set \begin{small}
\\$M\triangleq \sum_{k=1}^K a_k,$  
  $c_{1}\triangleq N_{0}^{M/2}\prod_{k=1}^{K}\left(\frac{b_{k}^{2}}{\gamma_\text{out} P_{L, k} \Gamma\left(a_{k}\right)^{-\frac{2}{a_{k}}}} \right)^{-\frac{a_{k}}{2}}$, $c_{2}\triangleq N_{u}e^{-K_{0}}\left(1+K_{0}\right) \gamma_\text{out}L$ and $d\triangleq [\frac{2c_{2}}{Mc_{1}}]^{\frac{-1}{M/2+1}}$ we get
  \begin{equation}
      P_\text{out}^{\infty}(E_{s},E_{u})=c_{1}E_{s}^{-M/2}+c_{2}E_{u}^{-1}.
  \end{equation} \end{small}
This function is convex. To find the minimum outage probability, let us define the Lagrangian $$\mathcal{L}\left(E_{S}, E_{u}\right)=c_{1}E_{s}^{-M/2}+c_{2}E_{u}^{-1}+\mu(E_{s}+E_{u}-E_{T}).$$
Where $\mu$ is the Lagrangian parameter.

By setting the partial derivatives of $P_\text{out}^{\infty}(E_{s},E_{u})$ to zero with respect to $E_{s},E_{u}$, and $\mu,$ we get the following system of equations:

\begin{align*}
  &E_s = \,\left(\frac{2 \mu}{Mc_{1}}\right)^{-\frac{1}{\frac{M}{2}+1}}, \\ 
  &E_u = \sqrt{\frac{c_2}{\mu}}, \\ 
  &E_s+E_u = \,E_T.
\end{align*}

This leads to the equation $E_s=d(E_T-E_s)^{\frac{4}{M+2}},$ which can be solved numerically. 

Therefore, if $x^{*}$ solves $x=d(E_T-x)^{\frac{4}{M+2}},$ then our desired solution is $E_s=x^{*}, E_u=E_{T}-x^{*}.$
\section{NUMERICAL RESULTS}
\label{Numerical Results}
In this section, we present numerical results corresponding to the considered multiple RISs-UAV system. We also validate the theoretical analysis using Monte-Carlo simulations. 

Unless otherwise stated, in all presented illustrations, the carrier frequency is assumed to be $f_{c} = 2$ GHz, the gains of the RISs in the first and second hops are, respectively given as $ G_{k,1} = G_{k,2} = 5$ dBi, and $\gamma_\text{out} = 0$ dB.
\begin{figure}[]
	\includegraphics[scale=0.4]{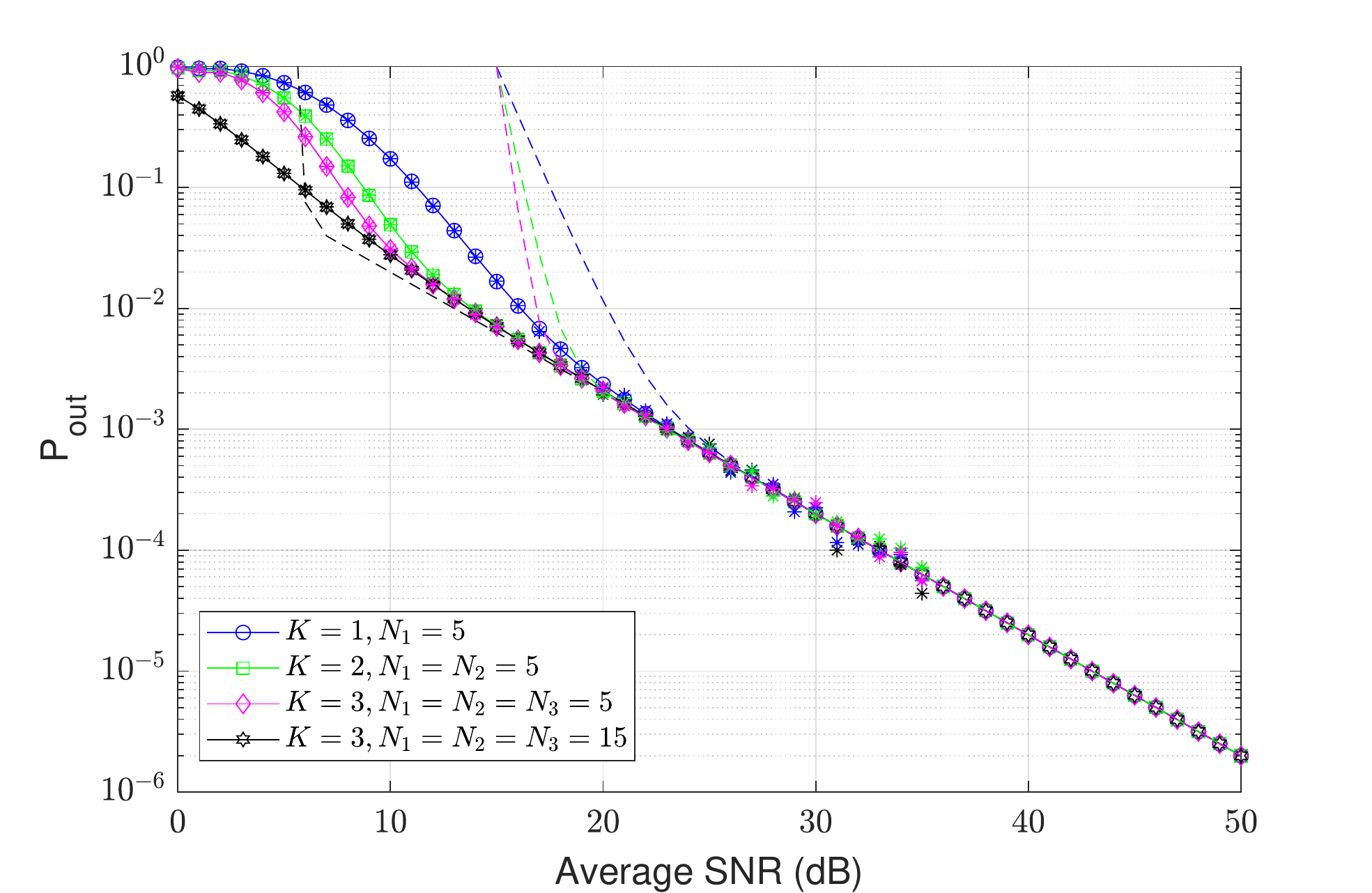}
	\centering
	\caption{The outage probability versus average SNR with different numbers of RIS $K\in\left\{1,2, 3\right\}$ and reflecting RIS elements $N_{1}, N_{2}, N_{3}\in\left\{5, 15\right\}$, $\left(m_{k,1}, m_{k,2}\right) = \left(1, 1\right)$, $\left(\Omega_{k,1}, \Omega_{k,2}\right) = \left(1,1\right),$ and $K_{0}=4.77 
 $ dB.}	
	\label{a}
\end{figure}

\begin{figure}[]
	\includegraphics[scale=0.4]{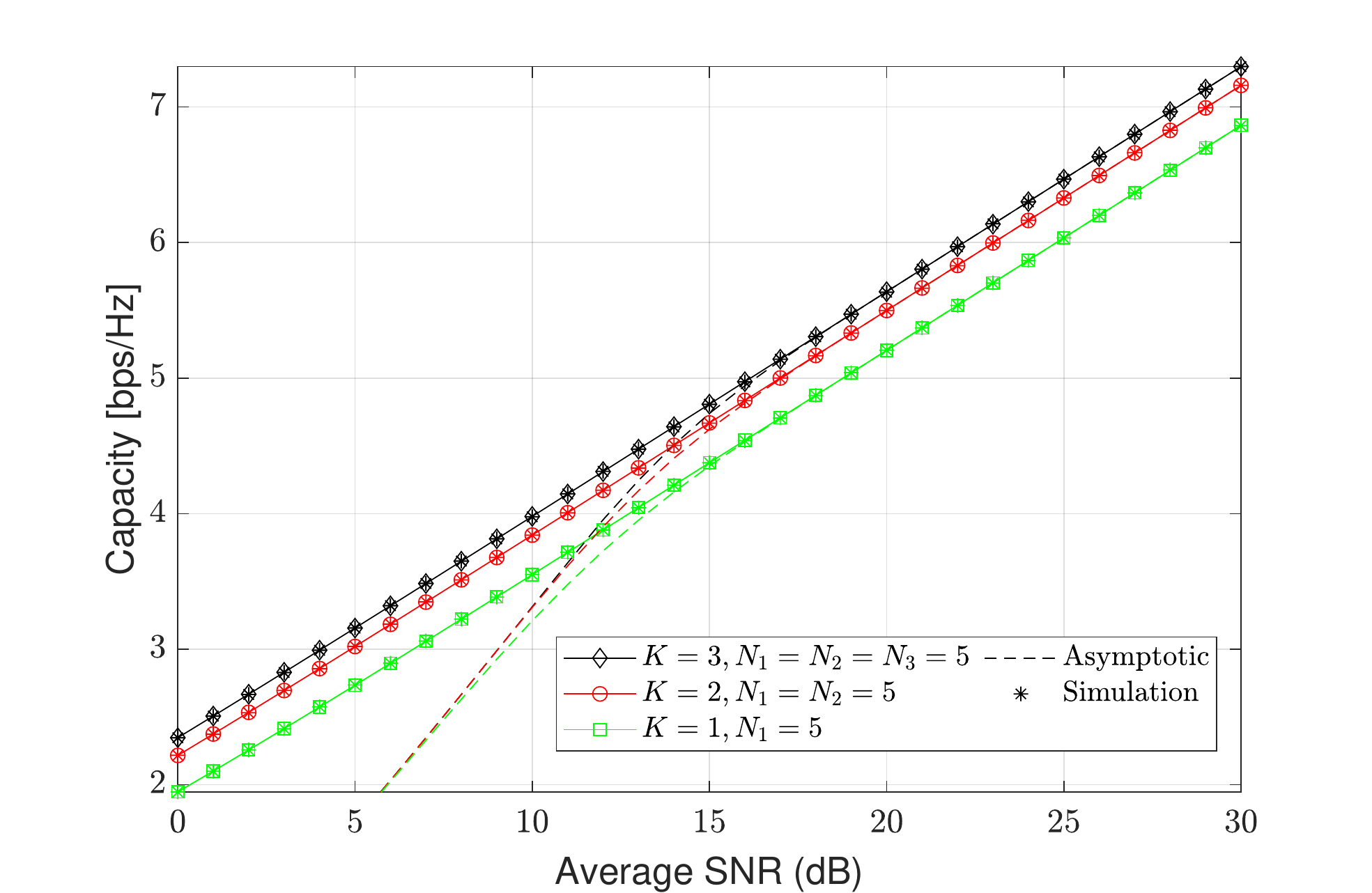}
	\centering
	\caption{Average channel capacity versus average SNR with different numbers of RIS $K\in\left\{1,2, 3\right\}$ and reflecting RIS elements $N_{1}, N_{2}, N_{3}\in\left\{5, 15\right\}$, $\left(m_{k,1}, m_{k,2}\right) = \left(1, 1\right)$, $\left(\Omega_{k,1}, \Omega_{k,2}\right) = \left(1,1\right),$ and $K_{0}=4.77 
 $ dB.}	
	\label{b}
\end{figure}

In Figures 1 and 2, we validate the analytical results by comparing them with simulations. In these two figures, the outage probability and average channel capacity are plotted against the average SNR. We can observe that simulations (asterisks) match the analytically derived results well. Furthermore, we can see that the asymptotic expressions, which are plotted as dashed lines, converge to the exact expressions at high SNR values. For the outage probability, it is obvious that all curves converge asymptotically to the same line. This happens when the second channel's outage probability, which decreases as $1/{\gamma}$ dominates the system performance. However, we see no such convergence at high SNR in the average channel capacity. It is clear from these two figures that a higher number of RISs and reflecting elements enhance the system performance, as expected.

In the next figures, we consider three RISs placed on a line perpendicular to the line between the source and the destination and placed 40 m away from the source.

\begin{figure}[]
	\includegraphics[scale=0.4]{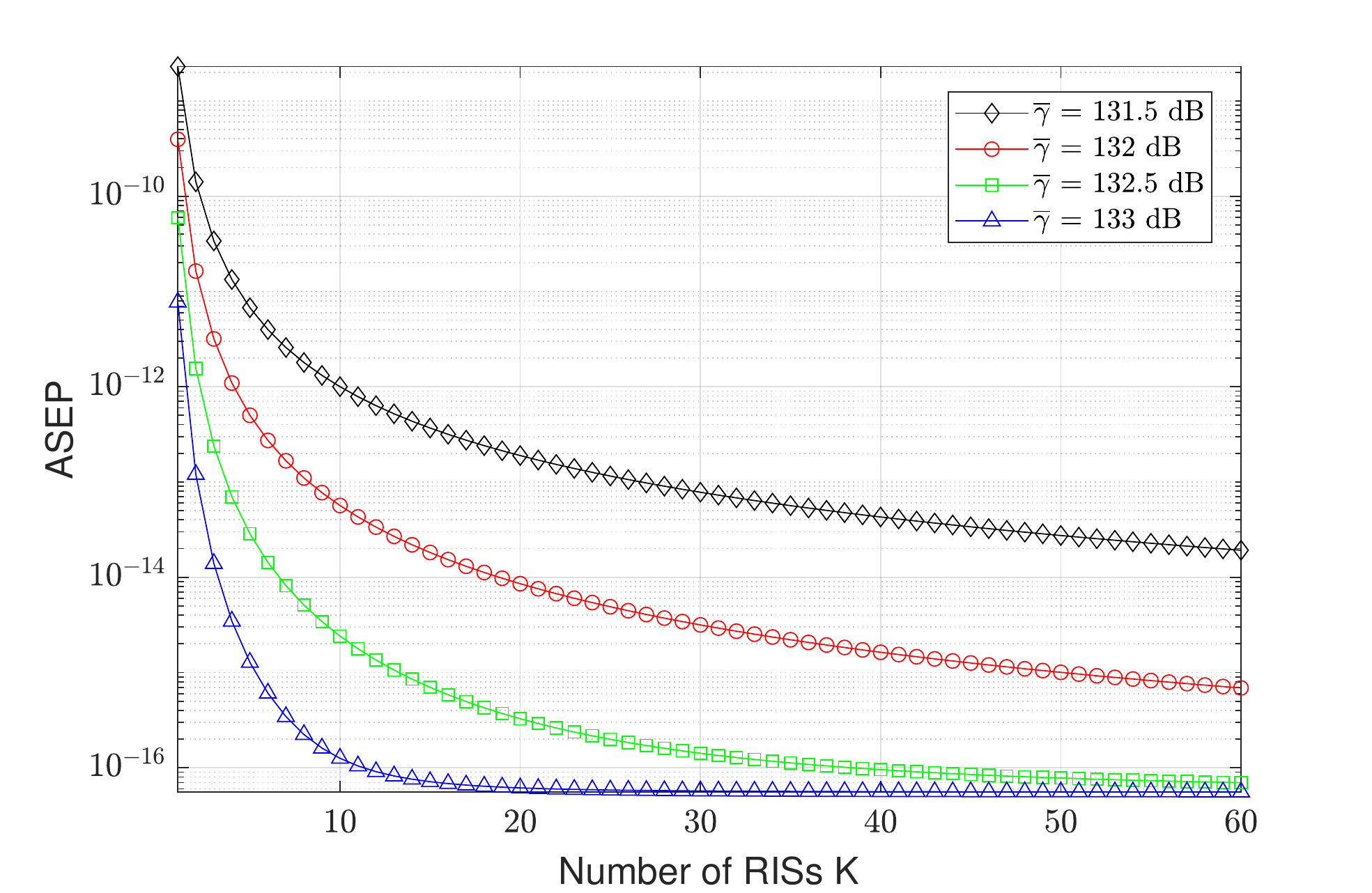}
	\centering
	\caption{Average symbol error probability versus the number of RISs for different average SNRs, 
 $\overline{\gamma_{a}} = \overline{\gamma_{b}} = \overline{\gamma},$ $K = 3,$ $N_{1} = N_{2}= N_{3} = 5$, $\left(m_{k,1}, m_{k,2}\right) = \left(1, 1\right)$, $\left(\Omega_{k,1}, \Omega_{k,2}\right) = \left(1, 1\right)$, and $K_{0}=4.77$ dB.}	
	\label{c}
\end{figure}

\begin{figure}[]
	\includegraphics[scale=0.4]{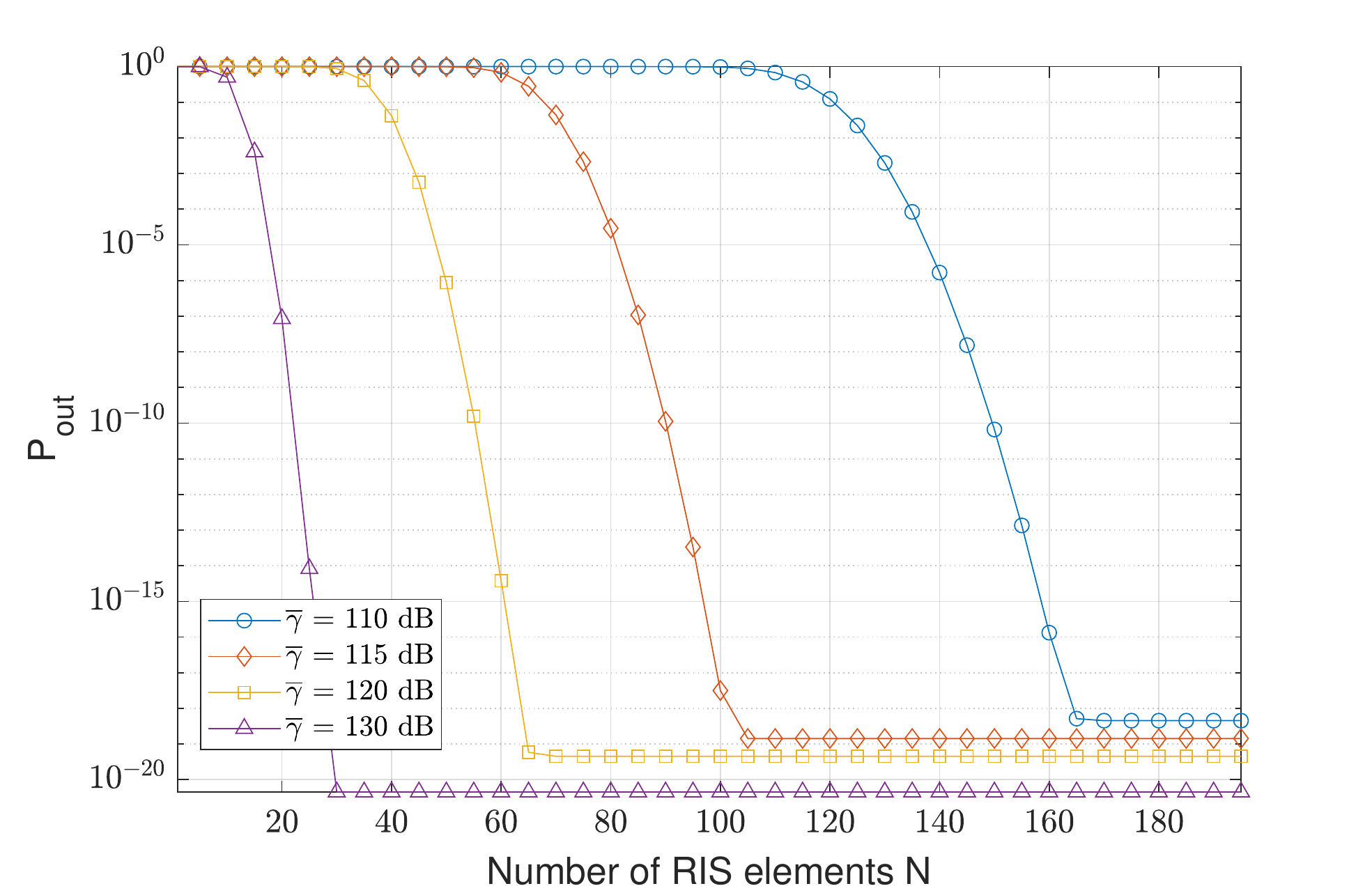}
	\centering
	\caption{Outage probability versus number of elements for different average SNRs, $\overline{\gamma_{a}} = \overline{\gamma_{b}} = \overline{\gamma},$ $K = 3,$ $N_{1} = N_{2}= N_{3} = 5$, $\left(m_{k,1}, m_{k,2}\right) = \left(1, 1\right)$, $\left(\Omega_{k,1}, \Omega_{k,2}\right) = \left(1, 1\right)$, and $K_{0}=4.77$ dB.}	
	\label{d}
\end{figure}
To further investigate the effect of the number of RISs $K$ and the number of their elements $N$ on the system, Figures 3 and 4 plot the average symbol error probability and outage probability against the number of RISs $K$ and the number of elements $N$, respectively. The figures show that the higher the number of RISs or reflecting elements, the better the performance. However, we can see that after a certain amount of RISs or reflecting elements, the average symbol error probability outage probability will be unaffected by changing the two parameters. This is because increasing $K$ and $N$ enhances the performance of the first hop but not the second hop, and so after some numbers, $K_{th}$ and $N_{th}$, the second hop will dominate the system. The numbers $K_{th}$ and $N_{th}$, as well as the outage limit, depend on the second hop and other system parameters such as $\overline{\gamma}$.

\begin{figure}[]
	\includegraphics[scale=0.4]{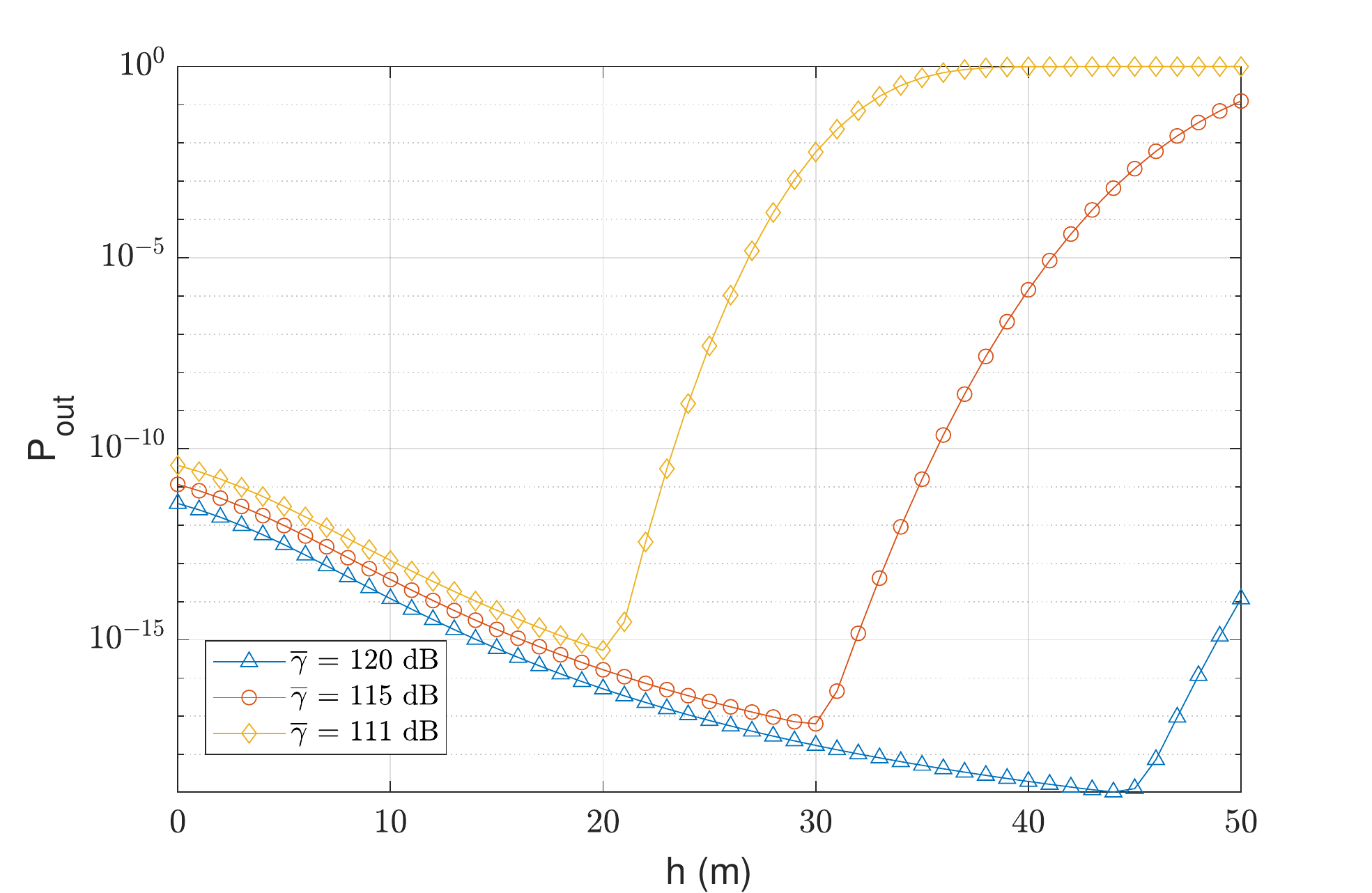}
	\centering
	\caption{The outage probability vs the UAV's height for different average SNRs, $\overline{\gamma_{a}} = \overline{\gamma_{b}} = \overline{\gamma},$ $K = 3,$ $N_{1} = N_{2}= N_{3} = 5$, $\left(m_{k,1}, m_{k,2}\right) = \left(1, 1\right)$, $\left(\Omega_{k,1}, \Omega_{k,2}\right) = \left(1, 1\right)$, and $K_{0}=4.77$ dB.}	
	\label{e}
\end{figure}
 Figure 5 demonstrates the effect of the UAV's height on the system. The outage probability is plotted against the height of the UAV, assuming an urban environment. It is clear that the outage probability decreases at first as the height increases. However, the outage probability starts increasing at some point, creating an optimal height. To see why we need to analyze the behavior of the first and the second hops separately. Initially, as the height increases, the Rician factor increases, hence, enhancing the second hop channel. But, for large elevation, the Rician factor will reach a limit, mainly $a_{2} \cdot e^{b_{2} \frac{\pi}{2}}$ and will not increase further. So the path-loss effect will dominate, and the second hop outage probability will increase. However, for the first hop, the outage probability increases with the height indefinitely. Thus, the optimal height of the system will be either at the point where the dominance shifts to the first hop or at the optimal height of the second hop. In these figures, the optimal height is at the shifts. Usually, the optimal height occurs at the point where the dominance shifts. This happens when the optimal point is lower than the second hop's optimal height, and the outage probability at this point is higher.
\begin{figure}[]
	\includegraphics[scale=0.4]{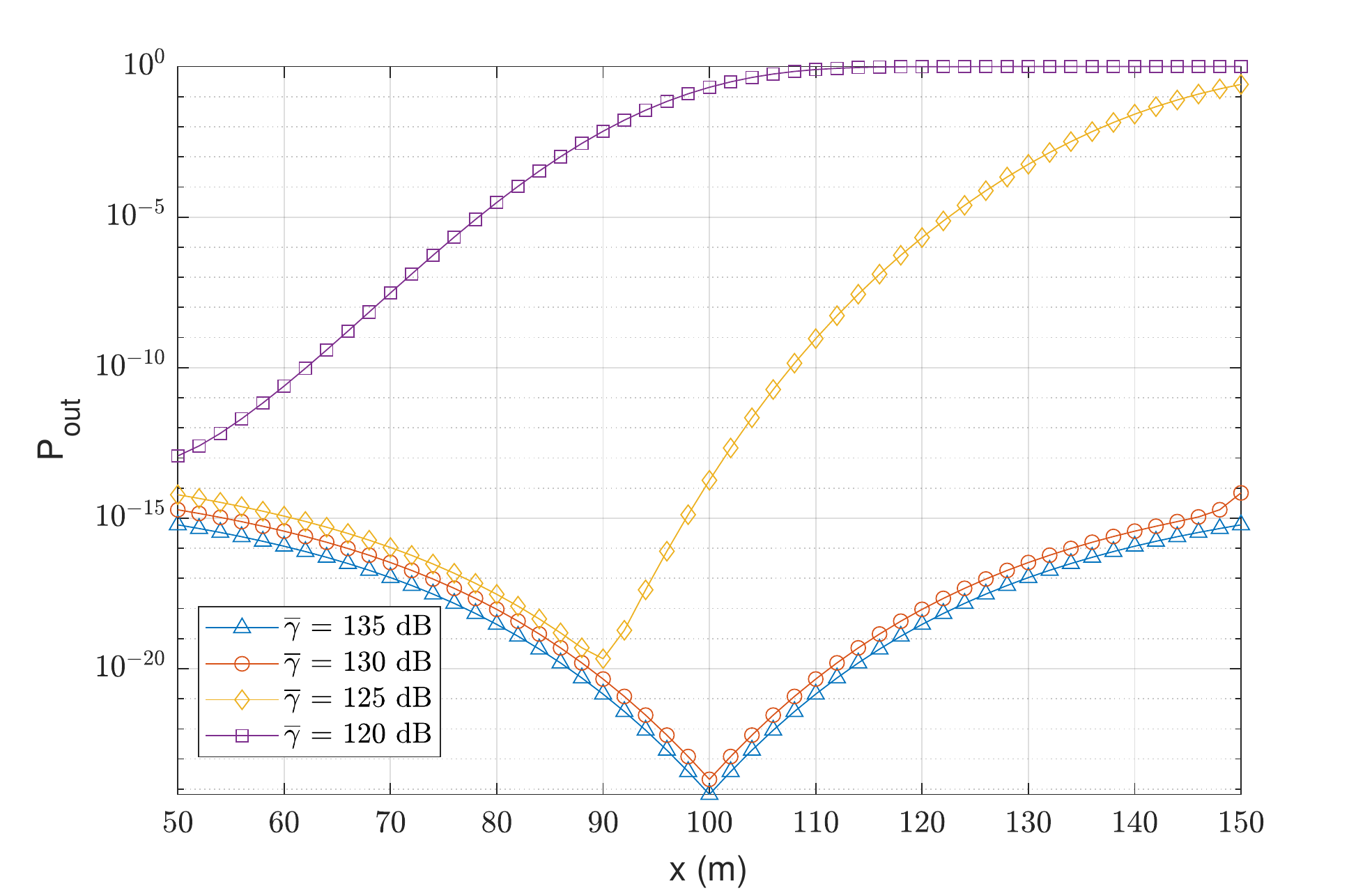}
	\centering
	\caption{The outage probability vs the UAV's horizontal position for different average SNRs, $\overline{\gamma_{a}} = \overline{\gamma_{b}} = \overline{\gamma},$ $K = 3,$ $N_{1} = N_{2}= N_{3} = 5$, $\left(m_{k,1}, m_{k,2}\right) = \left(1, 1\right)$, $\left(\Omega_{k,1}, \Omega_{k,2}\right) = \left(1, 1\right)$, and $K_{0}=4.77$ dB.}		
	\label{f}
\end{figure}
Figure 6 shows the behavior of the system when the UAV moves horizontally. Here, the destination is placed 100 m away from the source, and the UAV is at 50 m height. The outage probability is plotted against the horizontal distance of the UAV from the source. We can see that there is an optimal point where the outage is at its lowest. To understand this behavior, notice that moving away from the source increases the outage probability of the first hop, and moving away from the destination increases the outage probability of the second hop and vice versa. Thus, the first hop has an optimal position when the UAV is above the source, and the second hop has an optimal position when the UAV is above the destination. Therefore, the optimal position of the total system occurs either at the point where the dominance shifts (in diamonds) or at the source where the first hop is always dominant (in squares), or at the destination where the second hop is always dominant (in circles and triangles).
\begin{figure}[]
	\includegraphics[scale=0.4]{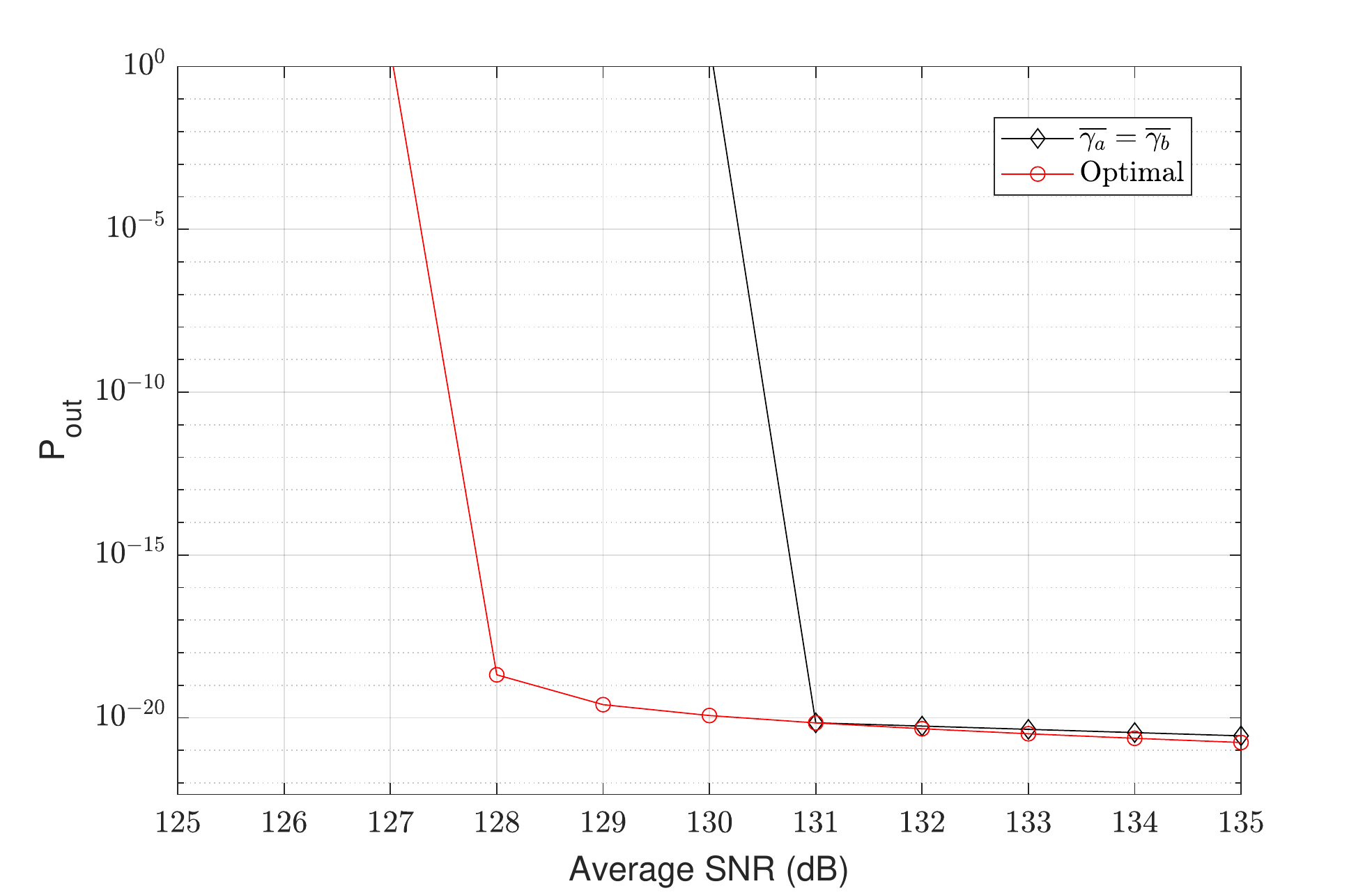}
	\centering
	\caption{The outage probability vs total average SNR $\overline{\gamma_{T}}$ for an optimal choice of $E_s,E_u$ and when $E_s=E_u$, $N{u}=N_{0},$ $K = 3,$ $N_{1} = N_{2}= N_{3} = 100$. $\left(m_{k,1}, m_{k,2}\right) = \left(1, 1\right)$, $\left(\Omega_{k,1}, \Omega_{k,2}\right) = \left(1, 1\right)$, and $K_{0}=4.77$ dB.}		
	\label{f}
\end{figure}

\begin{figure}[]
	\includegraphics[scale=0.51]{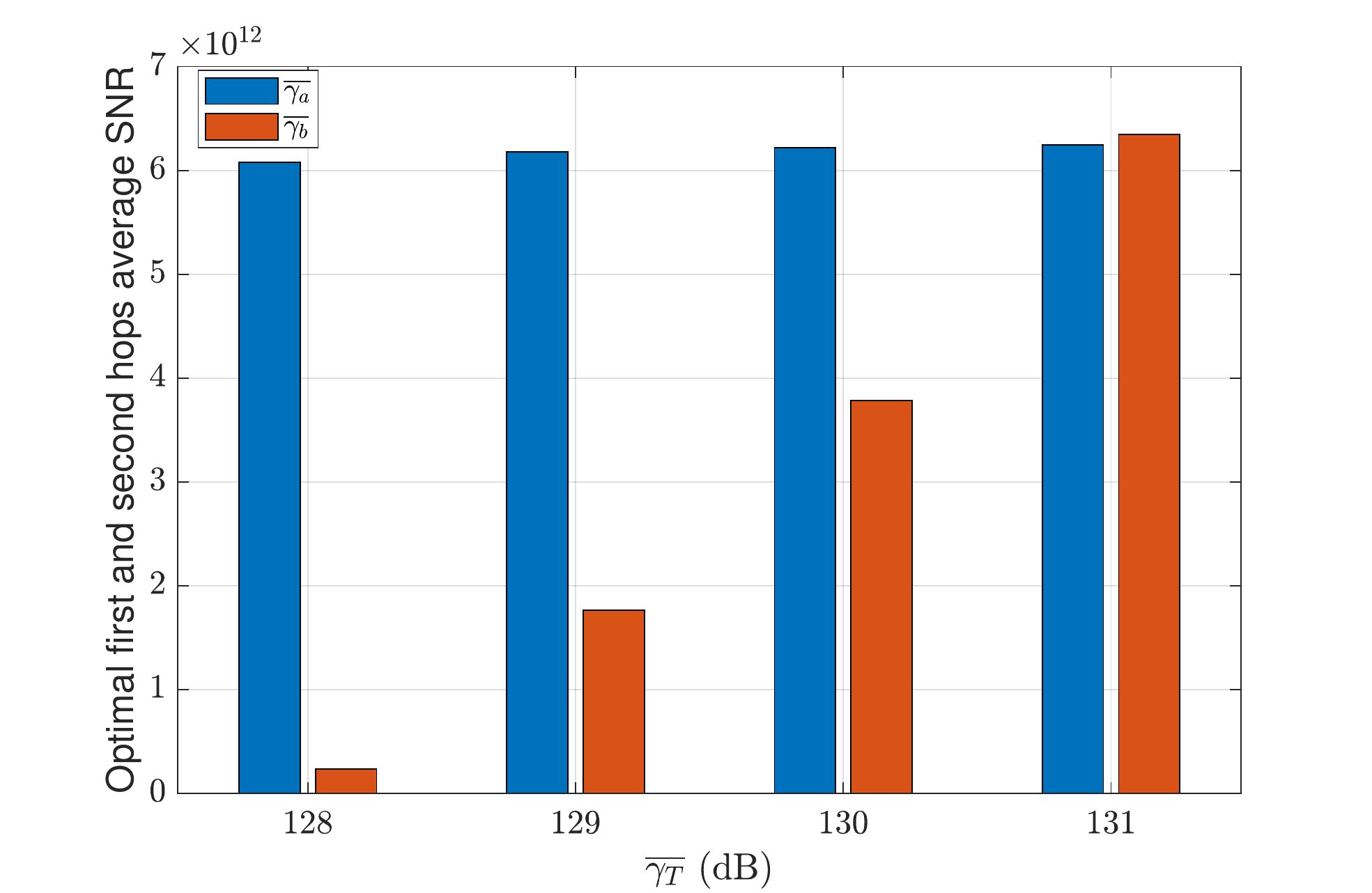}
	\centering
	\caption{The optimal channels average SNR vs total average SNR, $N{u}=N_{0},$ $K = 3,$ $N_{1} = N_{2}= N_{3} = 100$. $\left(m_{k,1}, m_{k,2}\right) = \left(1, 1\right)$, $\left(\Omega_{k,1}, \Omega_{k,2}\right) = \left(1, 1\right)$, and $K_{0}=4.77$ dB.}		
	\label{f}
\end{figure}
Figures 7 and 8 present the optimization results. Figure 7 plots outage probability against the total average SNR for the optimal and equal transmit power distribution. We see that our optimal solution gives a significant decrease in outage probability. We also see in figure 8 that the second hop share increases rapidly as the total power increases, in contrast to the first hop, where the increase is slow.

\section{CONCLUSION}\label{C}
In this paper, the performance of multiple reconfigurable intelligent surfaces-aided dual-hop UAV communication systems was studied over Nakagami-$m$ fading channels for the first hop and Rician fading channel for the second hop. Accurate closed-form approximations were first derived for each channel distribution and then used in deriving closed-form approximations for the outage probability, average symbol error probability, and average channel capacity assuming independent     non-identically distributed channels. Furthermore, an asymptotic expression was derived for the outage probability at the high signal-to-noise ratio regime to get more insights into the system performance. Results showed that multiple RISs could improve the performance and coverage of the UAV communication systems. Finally, we formulated and solved an optimization problem on the transmit power of each channel.
\begin{appendices}
\section{Proof of Corollary 2}
\end{appendices}
It is enough to show that $\int_0^t \frac{x^r}{1+x} d x = \sum_{k=0}^{\infty} \frac{(-1)^{k}}{r-k} t^{r-k}-\pi \csc (\pi r)$ for $r$ not an integer and $t>1$. Indeed, \begin{small} $$
\begin{aligned}
\int_0^t \frac{x^r}{1+x} d x&=\int_0^1 \frac{x^r}{1+x} d x+\int_1^t \frac{x^r}{1+x} d x \\
&=\int_0^1 \sum_{k=0}^{\infty}(-1)^k x^{r+k} d x+\int_1^t \sum_{k=0}^{\infty}(-1)^k x^{r-k-1} d x \\
&=\sum_{k=0}^{\infty} \frac{(-1)^k}{r+k+1}+\sum_{k=0}^{\infty} \frac{(-1)^k t^{r-k}}{r-k}-\sum_{k=0}^{\infty} \frac{(-1)^k}{r-k} \\
&=\sum_{k=0}^{\infty} \frac{(-1)^k t^{r-k}}{r-k}+\sum_{k=0}^{\infty}\left[\frac{(-1)^k}{r+k+1}-\frac{(-1)^k}{r-k}\right].
\end{aligned} 
.$$ \end{small} On the other hand 
\begin{align*}
-\frac{1}{r+1} &= \sum_{k=1}^{\infty}(-1)^k\left[\frac{1}{k+r+1}+\frac{1}{k+r}\right] \\
 &= \sum_{k=1}^{\infty}(-1)^k\left[\frac{1}{r+k+1}+\frac{2 r}{r^2-k^2}-\frac{1}{r-k}\right] 
 \end{align*}
  That is$$
-\frac{1}{r+1} = \sum_{k=1}^{\infty}(-1)^k\left[\frac{2 r}{r^2-k^2}\right]+\sum_{k=1}^{\infty}\left[\frac{(-1)^k}{r+k+1}-\frac{(-1)^k}{r-k}\right]
$$ so
\begin{align}\notag-\frac{1}{r}-\sum_{k=1}^{\infty}(-1)^k\left[\frac{2 r}{r^2-k^2}\right] &=\frac{1}{r+1}-\frac{1}{r}\\ \notag &+\sum_{k=1}^{\infty}\left[\frac{(-1)^k}{r+k+1}-\frac{(-1)^k}{r-k}\right].
\end{align}
Therefore $$
\sum_{k=0}^{\infty}\left[\frac{(-1)^k}{r+k+1}-\frac{(-1)^k}{r-k}\right]=-\frac{1}{r}-\sum_{k=1}^{\infty}(-1)^k\left[\frac{2 r}{r^2-k^2}\right].
$$
But, $-\frac{1}{r}-\sum_{k=1}^{\infty}(-1)^k\left[\frac{2 r}{r^2-k^2}\right]=-\pi \csc (\pi r)$\cite[Eq. (1.422.3)]{Grad.}, which completes the proof.

	\end{document}